\documentclass[journal,10pt]{IEEEtran}

\usepackage{amsmath}
\ifCLASSINFOpdf
  \usepackage[pdftex]{graphicx}
  \graphicspath{{../pdf/}{../jpeg/}}
  \DeclareGraphicsExtensions{.pdf,.jpeg,.png}
\else
  \usepackage[dvips]{graphicx}
  \DeclareGraphicsExtensions{.eps}
\fi

\usepackage{subfig}

\usepackage{amsmath,amssymb,amsthm,mathrsfs,dsfont}

\usepackage{cite}

\usepackage{tabularx}

\usepackage{color}

\newtheorem{lemma}{Lemma}
\newtheorem{proposition}{Proposition}

\newtheorem{definition}{Definition}

\begin{document}

\title{Incentive Mechanism Design for Wireless Energy Harvesting-Based Internet of Things}
%
%

\author{Zhanwei Hou, He Chen, Yonghui Li,~and Branka Vucetic

\thanks{Z.W. Hou, H. Chen, Y.H. Li and B. Vucetic are with the Centre of Excellence in Telecommunications, School of Electrical and Information Engineering, University of Sydney, Sydney, NSW, Australia (e-mail: \{zhanwei.hou, he.chen, yonghui.li, branka.vucetic\}@sydney.edu.au).}
}

\maketitle

\begin{abstract}
Radio frequency energy harvesting (RFEH) is a promising technology to charge unattended Internet of Things (IoT) low-power devices remotely. To enable this, in future IoT system, besides the traditional data access points (DAPs) for collecting data, energy access points (EAPs) should be deployed to charge IoT devices to maintain their sustainable operations. Practically, the DAPs and EAPs may be operated by different operators, and the DAPs thus need to provide effective incentives to motivate the surrounding EAPs to charge their associated IoT devices. Different from existing incentive schemes, we consider a practical scenario with asymmetric information, where the DAP is not aware of the channel conditions and energy costs of the EAPs. We first extend the existing Stackelberg game-based approach with complete information to the asymmetric information scenario, where the expected utility of the DAP is defined and maximized. To deal with asymmetric information more efficiently, we then develop a contract theory-based framework, where the optimal contract is derived to maximize the DAP's expected utility as well as the social welfare. Simulations show that information asymmetry leads to severe performance degradation for the Stackelberg game-based framework, while the proposed contract theory-based approach using asymmetric information outperforms the Stackelberg game-based method with complete information. This reveals that the performance of the considered system depends largely on the market structure (i.e., whether the EAPs are allowed to optimize their received power at the IoT devices with full freedom or not) than on the information availability (i.e., the complete or asymmetric information).
\end{abstract}


\section{Introduction}
\subsection{Background and Motivations}
By connecting objects, physical devices, vehicles, animals and other items to the Internet, Internet of Things (IoT) has emerged as a new paradigm to enable ubiquitous and pervasive communications \cite{Gubbi2013IoT,Lin2017Survey,Hou2017Recycling}. Wireless sensing service is one of the fundamental applications of IoT, which enables systems and users to continuously monitor ambient environment \cite{Peng2017Enhancing}.

One of the major hurdles for implementing the wireless sensing application is the limited lifetime of traditional battery-powered sensors, which are costly and hard to maintain \cite{kawabata2017robust,yang2017energy}. For example, frequent recharging or battery replacement is inconvenient in deserts or remote areas, and is even impossible for some scenarios, such as toxic environment or implanted medical applications \cite{Kama2015Wireless}. To tackle this problem, radio frequency energy harvesting (RFEH) has recently been proposed as an attractive technology to prolong the operational lifetime of sensors, enhance the deployment flexibility, and reduce the maintenance costs \cite{Kama2015Wireless,Niyato2016Novel}.

In this paper, we consider a radio frequency energy harvesting based IoT system consisting of a data access point (DAP) and several energy access points (EAPs). The DAP collects information from its associated sensors. EAPs can provide wireless charging services to sensors via the RF energy transfer technique. The sensors are assumed to have no embedded energy supply, but they can harvest energy from radio frequency (RF) signals radiated by the surrounding EAPs to transmit the data to the DAP \cite{chen2016cooperative}.

There are some research considering the deployment of dedicated EAPs in the existing cellular network, such that the upgraded network can provide both wireless access and wireless charging services \cite{huang2014enabling,huang2015cutting,zhong2015wireless,chen2017interference,chen2017stochastic,park2017energy,khan2017wirelessly,jiang2016secrecy}. However, it was assumed that the EAPs are deployed by the same operator of the existing network. In practice, the DAP and EAPs may be operated by different operators\footnote{This could happen when a resource-limited operator cannot provide radio frequency energy charging service in some certain area due to limited budgets, or lack of site locations, or lack of licensed spectrum for energy harvesting. Therefore, it has to resort to third-party operators.}. To effectively motivate these third-party and self-interested EAPs to help charge the sensors, effective incentive mechanisms are required to improve the payoff of the DAP as well as those of EAPs. While there are several initial work designing the incentive mechanism \cite{Chen2015stackelberg,sarma2016robust,ma2015distributed} for the EAPs belonging to different operators, complete information was considered in these schemes. Specifically, it was assumed that the EAPs will truthfully report their private information to the DAP, e.g., their energy costs and channel gains between EAPs and sensors. This happens when there exists a \emph{supervising entity} in the network, which is capable of monitoring and sharing all behaviours and network conditions of the DAP and EAPs to ensure that they always report the trustful information. However, without such a supervising entity, EAPs' private information might be not aware to the DAP, which is normally called \emph{information asymmetry} in the literature \cite{bolton2005contract}. A rational EAP may provide misleading information maliciously and pretend to be an EAP with better channel condition and/or higher energy cost to cheat for more rewards. A malicious EAP can succeed in cheating to get more benefits because of \emph{information asymmetry} in the RF energy trading process.

To address above issues, in this paper we will design effective incentive mechanisms to maximize the utilities of the DAP and EAPs under scenarios with asymmetric information. To this end, the following important questions should be addressed under asymmetric information:

\emph{Which EAPs the DAP should hire, how much energy should be requested from the hired EAPs, and how many rewards should be given to the hired EAPs?}

The above questions are non-trivial to answer because the hierarchical interactions between multiple parities should be modeled and analyzed: the cooperations between the DAP, the DAP's sensors and the EAPs, and the competition among EAPs with heterogeneous private information. Moreover, the information asymmetry make the problem even more challenging, because it is difficult for the DAP to hire the effective EAPs without knowing EAPs' private information, such as energy costs and channel condition towards its sensors.

\subsection{Solution and Contribution}
To answer the above questions, we apply the well-established economic theories to model the conflicted interests among the multiple parities in the considered RFEH-based IoT system. Specifically, we first extend the existing Stackelberg game-based approach with complete information to the considered case with asymmetric information, such that we can evaluate the performance degradation caused by information asymmetry to this approach. More specifically, due to lack of the complete information, the expected utility function of the DAP is defined and optimized in the Stackelberg game with asymmetric information. Considering that contract theory is a powerful tool originated from economics to deal with information asymmetry in a monopoly market, we apply contract theory to develop an optimal contract to effectively motivate the EAPs under asymmetric information. In our contract, the RF energy trading market is analogous as a monopoly labor market in economics. The DAP is modeled as the employer who offers a contract to each EAP. The contract is composed of a serious of contract items, which are combinations of energy-reward pairs. Each contract item is an agreement about how many rewards an EAP will get by contributing a certain amount of RF energy. Various heterogeneous EAPs are classified into different types according to their energy costs and instantaneous channel conditions. The EAPs are regarded as labors in the market, which will choose a contract item best meeting their interests. By properly designing the contract, an EAP's type will be revealed through its selection. Thus the DAP can capture each EAP's private information to a certain extent and thus relieve the issue of information asymmetry.

To the best knowledge, this is the first paper that systematically studies the RFEH-based IoT system under asymmetric information. The main contributions of this paper are summarized as follows.
\begin{itemize}
\item We develop the framework of RF energy trading in the RFEH-based IoT system and systematically design the incentive mechanisms for a practical scenario with asymmetric information.
\item To explore the performance degradtion due to lack of full information, we first extend the existing Stackelberg game-based approach to the considered case without instantaneous channel condition and energy costs of the EAPs by optimizing the expected utility function of the DAP. As contract theory is a powerful economic theory for designing incentive mechanism with asymmetric information. We then reformulate the problem by using contract theory. In our contract design, we characterize the necessary and sufficient conditions for the contract feasibility, i.e, individual rationality (IR) conditions and incentive capability (IC) conditions \cite{bolton2005contract}. Subject to the IR and IC constraints, the optimal contract under information asymmetry is achieved by maximizing the DAP's expected utility as well as the social welfare.
\item To compare the performance of the proposed contract theory-based approach using asymmetric information with that of the existing Stackelberg game-based method with complete information, we generalize the existing Stackelberg game formation with unified pricing to the case with discriminative pricing and derive the new Stackelberg equilibrium in closed-form. Here discriminative pricing means that we set different energy prices for different EAPs to fully exploit their potentials for the enery charging service. Numerical simulation results show that information asymmetry can lead to severe performance degradation for the Stackelberg game-based framework, while the proposed contract theory-based scheme using asymmetric information outperforms the Stackleberg game-based approach with complete information. This implies that \emph{the performance of the considered system depends largely on the market structure (i.e., whether the EAPs are allowed to optimize their received power at the IoT devices with full freedom or not) than on the information availability (i.e., the complete or asymmetric information)}.
\end{itemize}

Note that part of the work was presented in our previous conference paper \cite{hou2017contract}. In this journal version, we extend our previous work by considering both scenarios of complete information and asymmetric information and explore the impacts of information availability and market structure.

The rest of this paper is organized as follows: In  Section II, we review the related literature. The system model is presented in Section III. The incentive mechanisms in asymmetric information in complete information are proposed in Section IV. The benchmark schemes are elaborated in Section V. Numerical results are presented in Section VI, and conclusions are drawn in Section VII.

\section{Related Work}
\subsection{EAP assisted Wireless Energy Harvesting}
The idea of deploying a dedicated wireless energy network, that can provide wireless charging service to the terminals by using RFEH technology, was originally proposed by Huang \emph{et al.} \cite{huang2014enabling,huang2015cutting}. The dedicated power transmitters are called power beacons or EAPs. Using stochastic geometry, the tradeoff between the densities of the base stations and EAPs was analyzed in \cite{huang2014enabling}. There are many works exploiting EAPs to enable services of both wireless information and energy access in existing cellular networks\cite{zhong2015wireless,chen2017interference,chen2017stochastic,park2017energy,khan2017wirelessly,jiang2016secrecy}. Stochastic geometry was used to analyze the network performance with the EAPs in \cite{zhong2015wireless,chen2017interference,chen2017stochastic}. In \cite{park2017energy}, beamforming was introduced in the EAPs assisted cellular network to reduce the interference resulting from the EAPs. Leveraging finite-length information theory, the system performance in the finite blocklength regime was analyzed in \cite{khan2017wirelessly}. The security issue with EAP in the presence of a passive eavesdropper was investigated in \cite{jiang2016secrecy}. In all above works, the DAP and EAPs are assumed to belong to the same operator. In such a network, the devices belonging to the same operator with extra energy were assumed to voluntarily assist other devices. However, in practice, the DAP and EAPs may be operated by different operators. To successfully motivate self-interested EAPs to provide help, effective incentive mechanisms are required. There are several prior research works in designing the incentive mechanism\cite{Chen2015stackelberg,sarma2016robust,ma2015distributed} for the EAPs, where \cite{Chen2015stackelberg,sarma2016robust} adopted Stackelberg game and \cite{ma2015distributed} used auction to design the incentive mechanism. However, the existing incentive mechanisms only considered complete information scenario.

\subsection{Contract Theory}
The contract theory has been employed to address incentive design problems in wireless communication areas, such as mobile edge computing\cite{liu2017design}, device-to-device (D2D) communications\cite{Zhang2015Contract} and cooperative spectrum sharing\cite{Duan2014Cooperative}. To the best knowledge of the authors, we are the first to apply contract theory in the RF energy trading process in RFEH-based IoT systems. To design the incentive mechanism in such a scenario is challenging because the DAP needs to choose and reward the most efficient EAPs without knowing their channel conditions and energy costs.

\subsection{Stackelberg Game}
Stackelberg game has been widely used in wireless communications to model the interactions of steatitic parties, such as physical layer security\cite{zhang2012physical}, resource management for LTE-unlicensed\cite{zhang2017multi}, cognitive radio\cite{xu2014stackelberg} and wireless energy harvesting\cite{sarma2016robust,Chen2015stackelberg,yin2014optimal,chen2015distributed,Zhang2017Balancing}.
In \cite{yin2014optimal}, the authors considered cooperative spectrum sharing with one primary user (PU) and one secondary user (SU), which harvests energy from ambient radio signal. The Stackelberg game was used to design the the SU's optimal cooperation strategy. In \cite{chen2015distributed}, simultaneous wireless information and power transfer (SWIPT) in relay interference channels was considered, where multiple source-destination pairs communicate through their dedicated energy harvesting relays. The optimal power splitting ratios for all relays were derived by the formulated Stackelberg game. In \cite{Zhang2017Balancing}, the authors formulated a stochastic Stackelberg game to study the delay optimal power allocation scheme. There is a recent paper addressing the EAP assisted wireless energy harvesting by using Stackelberg game\cite{sarma2016robust}, but their system settings are different from those in our work. Only one EAP with multiple antennas was considered in this paper, and the EAP acts as the seller and the base station (BS) as the buyer in behalf of its sensors. A more relevant work is \cite{Chen2015stackelberg}, where an incentive mechanism was designed for the system with the similar setup where monetary reward with unified pricing were provided by the DAP to motivate third-party EAPs to assist the charging process. Here the unified pricing means prices per unit energy for different EAPs are the same. However, in this paper, we consider discriminative pricing scheme of Stackelberg game for the heterogenous EAPs in our work, including unified pricing scheme as a special case. Here discriminative pricing means that we set different energy prices for different EAPs to fully exploit their potentials for the enery charging service. Moreover, we extend the Stackelberg game to asymmetric information scenario by optimizing expected utility function of the DAP, instead of optimizing instantaneous utility function of the DAP in the classical Stackelberg game.

\section{System Model}
\begin{figure}
\centering \scalebox{0.45}{\includegraphics{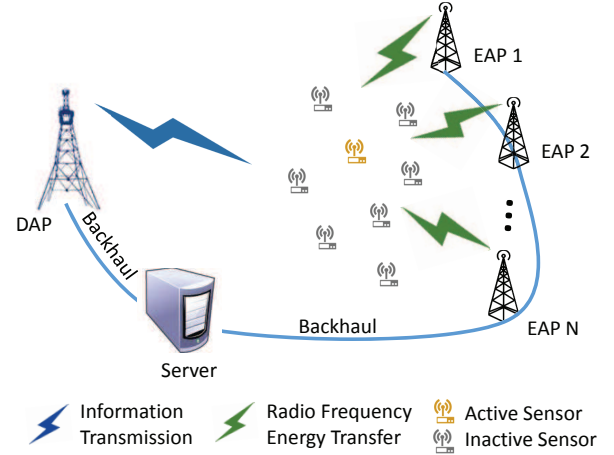}}
\caption{System model for the radio frequency energy harvesting assisted Internet of Things network \label{fig:system_model}}
\end{figure}
Consider a wireless energy harvesting-based IoT system consisting of one DAP and $N$ EAPs belonging to different operators, which are connected to constant power supplies and connected to the server by backhauls, as shown in Fig. \ref{fig:system_model}. The DAP is responsible for collecting various data from several wireless-powered sensors within its serving region. Without embedded energy supplies, the wireless-powered sensors fully rely on the energy harvested from the RF signals emitted by the EAPs to transmit its information to the DAP. For simplicity, we consider that the RF energy transfer and information transmission are performed over orthogonal bandwidth. For analytical tractability, time division-based transmission among sensors is adopted, i.e., there is only one active sensor during each transmission block. Hereafter, we refer to this active sensor as the information source. Besides, all the nodes in the system are assumed to be equipped with single antenna and operate in the half-duplex mode.

We consider that the energy-carrying signals sent by the EAPs are independent and identically distributed (i.i.d.) random variables with zero mean and unit variance. Note that no coordination between the EAPs is needed since independent signals are transmitted. All channels are assumed to experience independent slow and flat fading, where the channel gains remain constant during each transmission block and change independently from one block to another\footnote{Pilots are broadcasted by the active sensor to allow the DAP and EAPs to estimate the channels. So the DAP is aware of the channel gain from the DAP to the sensor and each EAP is aware of the channel gain from this EAP to the sensor. But the DAP generally is not aware of the channel gains from the EAPs to the active sensor. The energy consumption of channel estimation is ignored.}. The information source rectifies the RF signals received from the EAPs and uses the harvested energy to transmit its information. The time duration of every transmission block is normalized to one. So we use ``energy'' and ``power'' interchangeably hereafter. The amount of energy harvested by the information source during one transmission block can be expressed as
\begin{equation}\label{eq:Energy}
  E_{s} = \eta\sum\limits_{m=1}^{N} p_mG_{m,s},
\end{equation}
where $0<\eta<1$ is the energy harvesting efficiency, $p_m$ is the charging power of the $m$th EAP, and $G_{m,s}$ is the channel power gain between the $m$th EAP and the information source. Note that the noise is ignored in (\ref{eq:Energy}) since it is practically negligible at the energy receiver.

The harvest-use protocol is considered in this paper\cite{Krikidis2013harvest}. More specifically, the information source will use the harvested energy to perform instantaneous information transmission to the DAP. We consider a battery-free design which indicates that the sensor only has a storage device like supercapacitor to hold the harvested energy for a short period of time, e.g., among its scheduled transmission block. Hence the sensor exhausts all the harvested energy in each transmission block, so the sensor's energy storage device is emptied at the beginning of the transmission block. This battery-free design can reduce the complexity and costs of the sensors, which is particularly suitable for the considered IoT sensing applications and has been adopted by other applications \cite{Ju2014Throughput,Lu2015Wireless}. The transmit power of the information source is thus given by
\begin{equation}\label{eq:P_s}
  P_{s} = E_{s}.
\end{equation}
Then, the received signal-to-noise ratio (SNR) at the DAP is given by
\begin{equation}\label{eq:beta}
  \beta = \frac{p_s G_{a,s}}{N_0},
\end{equation}
where $N_0$ is the noise power at the DAP, and $G_{a,s}$ is the channel power gain from the information source to the DAP. Note that the time duration for each transmission block is normalized as one, such that the channel capacity and throughput can be used interchangeably. Hence the achievable throughput (bps) from the information source to the DAP can be expressed by
\begin{equation}\label{eq:Ras}
\begin{aligned}
R_{SD} &= W \log_2 (1+\beta) \\
&= W \log_2 \left( 1+\frac{\eta G_{a,s}}{N_0} \sum\limits_{m=1}^{N} p_m G_{m,s} \right),
\end{aligned}
\end{equation}
where $W$ is the bandwidth. We define the received signal power at the active sensor contributed by the $m$th EAP\footnote{Note that the received power contributed by each EAP is assumed to be distinguishable by considering that the EAPs work in disjoint narrow bandwidth. } as $q_m = p_m G_{m,s}$, and set $\gamma = \eta G_{a,s}/N_0$ for notation simplicity. We can thus simplify (4) as
\begin{equation}\label{eq:Ras2}
R_{SD} = W \log_2 \left( 1 + \gamma \sum\limits_{m=1}^{N} q_m \right).
\end{equation}

As we mentioned before, the EAPs considered in the system belong to different operators and act strategically, so they would not help the DAP voluntarily. To address this issue, the DAP needs to provide rewards to motive the EAPs to charge its sensors. In this paper, we mainly focus on monetary rewards as the incentive between operators. Other forms of rewards, such as physical resources (e.g., spectrum), or free offloading data between operator can also be used. To efficiently exploit the EAPs to achieve a good throughput, the following questions need to be answered in asymmetric information: \emph{Which EAPs the DAP should hire, how much energy should be requested from the hired EAPs, and how many rewards should be given to the hired EAPs?}

\section{Incentive Mechanisms with Asymmetric Information}
To answer the above questions in the practical scenario with asymmetric information, we first model the strategic interactions between the DAP and EAPs as a Stackelberg game. We will first re-design and re-analyze the existing Stackelberg game into the considered scenario by defining and optimizing expected utility of the DAP. In economic theories, contract theory is a powerful tool to design incentive mechanism in information asymmetry. As such, we will then reformulate the incentive mechanism problem into an optimal contract design problem.

\subsection{Stackelberg Game with Asymmetric Information}
In this part, we will first explore how to design a Stackelberg game to model the interactions between the DAP and the EAPs, and then derive the optimal energy prices under asymmetric information. In the proposed Stackelberg game with asymmetric information, the DAP provides rewards to the EAPs for charging its sensors. The DAP is the leader of the formulated Stackelberg game, which imposes energy prices for the EAPs. The DAP optimizes the energy prices to maximize its expected utility function defined as the difference between the benefits obtained from the achievable throughput and its total payment to the EAPs. The EAPs are the followers which optimize their utility functions defined as the payment received from the DAP minus its energy cost.

\subsubsection{Stackelberg Game Formulation}
The channel conditions and energy costs of EAPs are different, so the efficiencies of EAPs to charge the sensor are distinct. To fully exploit the potential of the EAPs, a discriminative pricing strategy is considered, i.e., the DAP can impose different prices of per unit energy harvested from different EAPs. Let $\pmb q = [q_1,q_2,\dots,q_N]^T$ as the vector of the active sensor's received power from EAPs, with $q_m$ denoting the received power from the $m$th EAP, and let $\pmb \lambda = [\lambda_1,\lambda_2,\dots,\lambda_N]^T$ as the vector of prices per unit energy harvested from EAPs, with $\lambda_m \ge 0$ denoting the price per unit energy harvested from the $m$th EAP. The total payment of the DAP to the EAPs is
\begin{equation}
\Lambda(\pmb \lambda,\pmb q)=\sum\limits_{m=1}^{N} \lambda_m q_m,
\end{equation}
where $q_m$ is the received energy from $m$th EAP. Since the aim of the DAP is to achieve higher throughput at the cost of less rewards to the EAPs, the utility function of the DAP can be defined as
\begin{equation}
U_{DAP}^S(\pmb \lambda,\pmb q)=R_{SD} - c \Lambda(\pmb \lambda,\pmb q).
\end{equation}
where $R_{SD}$ is the achievable throughput defined in (\ref{eq:Ras}) and (\ref{eq:Ras2}), $c$ is the unit cost of the DAP, which is normalized as $c = 1$ without loss of generality hereafter.

Each EAP is modeled as a follower which would like to maximize its individual profit, the utility of which is defined as
\begin{equation}\label{eq:USEAP_Asym}
U_k^S(\lambda_k,q_k) = \lambda_k q_k - \mathcal{C}_k (p_k),
\end{equation}
where $p_k = q_k /G_{k,s}$ is the transmit power of the $k$th EAP, and $\mathcal{C}_k(\cdot)$ is used to model the energy cost of the $k$th EAP, given by
\begin{equation}
\mathcal{C}_k (x) = a_k x^2,
\end{equation}
where $a_k > 0$ is the energy cost coefficient. Note that the above quadratic function has been widely adopted in the energy trading market to model the energy cost\cite{Mohsenian2010Autonomous}. The utility function of the $k$th EAP becomes
\begin{equation}\label{eq:USEAP_theta}
U_k^S(\lambda_k,q_k) = \lambda_k q_k - \frac{a_k}{G_{k,s}^2}q_k^2,
\end{equation}

Since the DAP is not aware of each EAP's exact energy cost coefficient and channel gain, it can sort EAPs into some discrete types and use the statistical distributions of the types of EAPs from historical data to optimize the expected utility of the DAP. Specifically, we define the type of the $k$th EAP as
\begin{equation}\label{eq:theta}
\theta_k := \frac{G_{k,s}^2}{a_k},
\end{equation}
which suggests that the larger the channel gain $G_{k,s}$ between the EAP and the information source, and/or the lower the unit energy cost coefficient $a_k$, the higher the type of the EAP. Without loss of generality, we assume that there are totally $K$ types of EAPs with $\theta_1 < \theta_2 < \dots < \theta_K$. In this definition, the higher type EAP has better channel quality and/or lower energy cost coefficient. Note that since $a_k>0$ and $G_{k,s}>0$, $\theta >0$ holds. Using (\ref{eq:theta}), the EAP's utility can be rewritten as
\begin{equation}\label{eq:UEAPSim}
U_k^S(\lambda_k,q_k) = \lambda_k q_k - \frac{q_k^2}{\theta_k}.
\end{equation}
Assume there are $N_k$ EAPs belonging to the $k$th type, we thus have $\sum\nolimits_{k=1}^{K} N_k = N$. We then can rewrite the DAP's utility according to the types of EAPs as
\begin{equation}\label{eq:UDAPNew}
U_{DAP}^S(\lambda_k,q_k) = W \log_2 \left( 1 + \gamma \sum\limits_{k=1}^{K} N_k q_k \right) - \sum\limits_{k=1}^K N_k \lambda_k q_k.
\end{equation}

In this section, we consider a scenario with strong information asymmetry. In such a scenario, the DAP is only aware of the total number of EAPs (i.e., $N$) and the distribution of each type. But it does not know each EAP's private type and thus it does not know the exact number of EAPs belonging to each type $k$ (i.e., $N_k$). As such, the DAP needs to optimize its expected utility over the possibilities of all possible combinations of $N_k$. The expected utility of the DAP with $N$ EAPs is given by
\begin{equation}\label{eq:Expect_U_DAP}
\begin{aligned}
  &\mathbb{E}\{U_{DAP}^S(\pmb \lambda,\pmb q)\} = \sum\limits_{n_1=0}^{N}\sum\limits_{n_2=0}^{N-n_1}\dots \sum\limits_{n_{K-1}=0}^{N-\sum\nolimits_{i=0}^{K-2}n_i}\\
  &\left\{ \Phi_{n_1,\dots,n_K} \left[ W \log_2 \left( 1 + \gamma \sum\limits_{k=1}^{K} n_k q_k \right) - \sum\limits_{k=1}^K n_k \lambda_k q_k \right] \right\},
\end{aligned}
\end{equation}
where $n_K = N - \sum\nolimits_{i=0}^{K-1}n_i$ is known after giving $n_1,n_2,\dots,n_{K-1}$ since the DAP knows the total number $N$ of EAPs, and $\Phi_{n_1,\dots,n_K}$ is the probability of a certain combination of the number of EAPs belonging to each type (i.e., $N_k,\{k=1,2,\dots,K\}$). We assume that all types are uniformly distributed. The probability of one EAP belonging to each type is the same, which is $1/K$. In this case, $\Phi_{n_1,\dots,n_K}$ can be calculated as
\begin{equation}\label{eq:Prob_combine}
\begin{aligned}
  \Phi_{n_1,\dots,n_K}&=\textbf{Pr} \left( N_1=n_1, N_2=n_2, \dots, N_{k}=n_{k} \right)\\
  &=\frac{N!}{n_1! n_2! \dots n_K! K^N}
\end{aligned}
\end{equation}

Since the DAP is not aware of the EAPs' private information, it can only optimize the expectation of DAP's utility function by using the statistical knowledge of the EAP's private information. So the optimization problem for the DAP or the leader-level game can be formulated as
\begin{equation}\label{eq:opt_DAP_asym}
\begin{aligned}
  \mathbf{(P4.1):} \quad & \max\limits_{\pmb \lambda} \mathbb{E}\{U_{DAP}^S(\pmb \lambda,\pmb q)\} \\
  s.t. \quad & \pmb \lambda \ge \pmb 0
\end{aligned}
\end{equation}
Accordingly, the optimization problem for the EAP with $k$th type or the follower-level game can be formulated as
\begin{equation}\label{eq:opt_EAP_asym}
\begin{aligned}
  \mathbf{(P4.2):} \quad & \max\limits_{q_k} U_k^S(\lambda_k,q_k) \\
  s.t. \quad & q_k \ge 0 \quad\quad\quad\quad\quad\quad
\end{aligned}
\end{equation}
Note that although the DAP does not know the EAP's exact type, it knows the type set of EAPs.

The Stackelberg game for the considered system has been formulated by combining problems (P4.1) and (P4.2). In this game, the DAP is the leader who aims to solve problem (P4.1), while the EAPs are the followers who aim to solve their individual problem (P4.2). Once a game is formulated, the subsequent task is to find its equilibrium point(s). For the solution of the formulated game, the most well-known concept is the Stackelberg equilibrium (SE), which can be formally defined as follows:

\begin{definition}[\textbf{Stackelberg equilibrium (SE)}]
We use $\pmb \lambda^* = [\lambda_1^*,\lambda_2^*,\dots,\lambda_N^*]^T$ and $\pmb q^* = [q_1^*,q_2^*,\dots,q_N^*]^T$ to denote the solutions of problems (P4.1) and (P4.2), respectively. Then, $\left({\pmb \lambda}^*,{\pmb q}^* \right)$ is a SE of the formulated game if the following conditions are satisfied
\begin{equation}\label{eq:def_SE_1}
{U_{DAP}^S}\left( {{\pmb \lambda}^* ,{\pmb q}^*} \right) \ge
{U_{DAP}^S}\left( {{\pmb \lambda} ,{\pmb q}^*} \right),
\end{equation}
\begin{equation}\label{eq:def_SE_2}
{U_m^S}\left( {{q_m^*},{\lambda_m^*}} \right) \ge
{U_m^S}\left( {{q_m},{\lambda_m^*}} \right),
\end{equation}
for all $\pmb \lambda \ge  \pmb 0$ and $\pmb q \ge \pmb 0$.
\end{definition}

\subsubsection{Analysis of the Proposed Game}
In this part, we will analyze the SE of the proposed Stackelberg game with asymmetric information.

It can be observed from (\ref{eq:USEAP_theta}) that for given values of $\lambda_k$, the utility function of the the EAP with the $k$th type is a quadratic function of its contributed power $q_k$ to the active sensor and the constraint is affine, which indicates that the problem (P4.2) is a convex optimization problem. Thus, it is straightforward to obtain its optimal solution given in the following lemma:
\begin{lemma}\label{lemma:EAP_optimal_solution_asym}
For given values of $\lambda_k$, the optimal $q_k^*$ of the EAP with $k$th type for problem (P4.2) is given by
\begin{align}\label{eq:EAP_optimal_solution_asym}
{q_k^{{*}}} = \frac{\theta_k \lambda_k}{2}.
\end{align}
\end{lemma}
\begin{proof}
The proof of this lemma follows by noting that the objective function of problem (P4.2) given in (\ref{eq:opt_EAP_asym}) is a quadratic function in terms of $q_k$.
\end{proof}
It can be observed from Lemma \ref{lemma:EAP_optimal_solution_asym} that for the same energy price, an EAP with better channel gain and/or less energy cost would like to contribute more power to the sensors.

Then we replace $q_k$ with $q_k^*$ in problem (P4.1), the optimization problem at the DAP side can be expressed as
\begin{equation}\label{eq:opt_DAP_uniform_asym}
\begin{aligned}
  \mathbf{(P4.3):} \quad & \max\limits_{\pmb \lambda} \mathbb{E}\{U_{DAP}^S(\pmb \lambda,\pmb q^*)\} \\
  s.t. \quad & \pmb \lambda \ge \pmb 0
\end{aligned}
\end{equation}
where $\mathbb{E}\{U_{DAP}^S(\pmb \lambda,\pmb q^*)\}$ is given by
\begin{equation}\label{eq:USEAP_Bar_asym}
\begin{aligned}
  &\mathbb{E}\{U_{DAP}^S(\pmb \lambda,\pmb q^*)\} = \sum\limits_{n_1=0}^{N}\sum\limits_{n_2=0}^{N-n_1}\dots \sum\limits_{n_{K-1}=0}^{N-\sum\nolimits_{i=0}^{K-2}n_i}\\
  &\left\{ \Phi_{n_1,\dots,n_K} \left[ W \log_2 \left( 1 + \frac{\gamma}{2} \sum\limits_{k=1}^{K} n_k \theta_k \lambda_k \right) \right. \right.  \\
  &\left. \left. - \frac{1}{2}\sum\limits_{k=1}^K n_k \theta_k \lambda_k^2 \right] \right\},
\end{aligned}
\end{equation}
where $\Phi_{n_1,\dots,n_K}$ is given in (\ref{eq:Prob_combine}).

We can observe that problem (P4.3) is a concave function in terms of vector $\pmb \lambda$. This is because each term in the summation is composed by a logarithm function (concave) and quadratic functions (concave), and the summation of concave functions are still a concave function. Moreover, the constraint is affine. Problem (P4.3) is then a convex optimization problem. So we can numerically solve the system of equations given by the KKT conditions to get the solution of problem (P4.3). According the KKT conditions, we can also get some insight about the structure of the solution and thus have the following proposition.
\begin{proposition}\label{proposition:KKT_solution_insight}
The optimal solution to problem (P4.3) have the following structure:
\begin{equation}\label{eq:KKT_sol_insight}
\lambda_1^* = \lambda_2^* = \dots = \lambda_N^*.
\end{equation}
\end{proposition}
\begin{proof}
See Appendix \ref{appendix:Prop_2}.
\end{proof}

We surprisingly find that the optimal energy prices for different EAPs are the same, even if we impose discriminative prices for different EAPs in the original design of the Stackelberg game. This is because the energy price of unit received power is used in our pricing scheme. The DAP has no motivation to treat the received power from EAPs differently, so a unified pricing per unit received power is achieved.

Lack of complete information, the performance of the Stackelberg game with asymmetric information is worse than that with complete information. Note that in the considered scenario, there are $N$ EAPs in the market. In each channel realization, each EAP in the market selects one EAP type from a EAP type set randomly. In each channel realization, the Stackelberg game under complete information can adapt to the instantaneous combination of EAP types and calculate an optimal price for each instantaneous combination of EAP types by optimizing the instantaneous utility function. While the Stackelberg game under asymmetric information cannot adapt to instantaneous combination of EAP types, since it can only calculate a single price for all possible combinations of EAP types. Therefore, the reason that Stackelberg game under asymmetric information is worse than Stackelberg game under complete information is that it fails to adapt to the change of the instantaneous combinations of EAP types, i.e., the change of wireless channel conditions. This deduction will be verified later in the simulation part.

\subsection{Optimal Contract with Asymmetric Information}
As we mentioned above, the performance of the Stackelberg game is degraded under asymmetric information. To improve the performance under asymmetric information, the DAP could design and offer a contract to effectively motivate the EAPs to charge its sensors. Note that in Stackelberg game, the EAP has the freedom to optimize its own utility by choosing any amount of received signal power at the active sensor when the DAP imposes some given energy price. Different from Stackelberg game, limited options are allowed for EAPs to select in contract theory. Specifically, a group of energy-reward pairs (referred to as contract items) are designed. A contract consisting of a group of contract items is provided to the EAPs. The EAPs will choose a contract item at its discretion to maximize its benefit. By properly designing the contract item, the DAP can induce the EAP to expose its type by its selection of the contract item and thus relieve the information asymmetry.

In the following, we will formulate the optimal contract, characterize its feasibility conditions and provide optimal solution for the formulated contract.

\subsubsection{Contract Formulation}
In this part, we will formulate a contract for the RF energy trading between the DAP and EAPs, characterize its feasibility conditions, and derive the optimal contract subject to the feasibility conditions.

A contract including a series of energy-reward pairs $(q_k, \pi_k)$ is designed to maximize the expectation of the DAP's utility. For the $k$th type EAP, $q_k$ is the received power contributed by $k$th EAP and $\pi_k$ is the reward paid to the $k$th EAP as the incentive for the corresponding contribution.

We first rewrite the utility functions of the DAP and EAPs according to contract items. The DAP's utility function is thus given by
\begin{equation}\label{eq:UDAPNew_contract}
U_{DAP}^C(\pmb \pi,\pmb q) = W \log_2 \left( 1 + \gamma \sum\limits_{k=1}^{K} N_k q_k \right) - \sum\limits_{k=1}^K N_k \pi_k,
\end{equation}
where $\pmb \pi=[\pi_1,\pi_2,\dots,\pi_K]^T$ is the reward paid by the DAP to the EAP with the $k$th type for its corresponding contribution $\pmb q=[q_1,q_2,\dots,q_K]^T$. Similar to (\ref{eq:Expect_U_DAP}), the expectation of $U_{DAP}^C(\pmb \pi,\pmb q)$ can be represented as
\begin{equation}\label{eq:Expect_U_DAP_contract}
\begin{aligned}
  &\mathbb{E}\{U_{DAP}^C(\pmb \pi,\pmb q)\} = \sum\limits_{n_1=0}^{N}\sum\limits_{n_2=0}^{N-n_1}\dots \sum\limits_{n_{K-1}=0}^{N-\sum\nolimits_{i=0}^{K-2}n_i}\\
  &\left\{ \Phi_{n_1,\dots,n_K} \left[ W \log_2 \left( 1 + \gamma \sum\limits_{k=1}^{K} n_k q_k \right) - \sum\limits_{k=1}^K n_k \pi_k \right] \right\},
\end{aligned}
\end{equation}
where $n_K = N - \sum\nolimits_{i=0}^{K-1}n_i$ is known after giving $n_1,n_2,\dots,n_{K-1}$ since the DAP knows the total number $N$ of EAPs, and $\Phi_{n_1,\dots,n_K}$ is the probability of a certain combination of the number of EAPs belonging to each type (i.e., $N_k,\{k=1,2,\dots,K\}$), which is given by (\ref{eq:Prob_combine}). And then the utility function of the EAP with the $k$th type is rewritten as
\begin{equation}\label{eq:UEAPSim_contract}
U_k^C(\pi_k,q_k) = \pi_k - \frac{q_k^2}{\theta_k}.
\end{equation}
The social welfare is defined as the summation of the utilities of the DAP and all $N$ EAPs, given by
\begin{equation}\label{eq:SocWel}
\begin{aligned}
  \Gamma(\pmb \pi,\pmb q) &= U_{DAP}^C(\pmb \pi,\pmb q) + \sum\limits_{k=1}^{K} N_k U_k^C(\pi_k,q_k) \\
  &= W \log_2 \left( 1+ \gamma \sum\limits_{k=1}^K N_k q_k \right) - \sum\limits_{k=1}^K \frac{N_k q_k^2}{\theta_k}.
\end{aligned}
\end{equation}
It can be seen that the internal transfers, i.e., rewards, are cancelled in the social welfare, which is consistent with the aim to maximize the efficiency of the whole system, i.e., achieving more throughput at the cost of less energy consumptions.

Next, we will figure out the feasibility conditions. In our design, to encourage the EAPs to participate in the charging process and ensure that each EAP only chooses the contract item designed for its type, the following individual rationality (IR) and incentive compatibility (IC) constraints should be satisfied \cite{bolton2005contract}.
\begin{definition}[\textbf{Individual Rationality (IR)}]
 The contract item that an EAP chooses should ensure a nonnegative utility, i.e.,
\begin{equation}\label{eq:IR}
U_k^C(\pi_k,q_k) = \pi_k - \frac{q_k^2}{\theta_k} \ge 0, \forall k \in \{1,\dots,K\}.
\end{equation}
\end{definition}

\begin{definition}[\textbf{Incentive Compatibility (IC)}]
An EAP of any type $k$ prefers to choose the contract item $(q_k, \pi_k)$ designed for its type, instead of any other contract item $(q_j, \pi_j), \forall j \in \{1,\dots,K\}$ and $j \ne k$, given by
\begin{equation}\label{eq:IC}
  \pi_k - \frac{q_k^2}{\theta_k} \ge \pi_j - \frac{q_j^2}{\theta_k}, \forall k,j \in \{1,\dots,K\}.
\end{equation}
\end{definition}

The IR condition requires that the received reward of each EAP should compensate the cost of its consumed energy when it participates in the energy trading. If $U_k \le 0$, the EAP will choose not to charge the information source for the DAP. We define this case as $(q_k = 0, \pi_k = 0)$. The IC condition ensures that each EAP automatically selects the contract item designed for its corresponding type. The type of each EAP is thus revealed to the DAP, which is called ``self-reveal''. If a contract satisfies the IR and IC constraints, we refer to the contract as a feasible contract.

Following the idea of contract theory\cite{bolton2005contract}, the DAP aims at maximizing its expected utility subjecting to the constraints of IR and IC given in (\ref{eq:IR}) and (\ref{eq:IC}). Thus, the optimal contract is the solution to the following optimization problem
\begin{equation}\label{eq:CntFrm}
\begin{aligned}
   \mathbf{(P4.4):} \qquad & \max\limits_{(\pmb \pi,\pmb q)}  \mathbb{E}\{U_{DAP}^C(\pmb \pi,\pmb q)\} \\
  s.t. \quad & \pi_k - \frac{q_k^2}{\theta_k} \ge 0, \forall k \in \{1,\dots,K\}, \\
  &\pi_k - \frac{q_k^2}{\theta_k} \ge \pi_j - \frac{q_j^2}{\theta_k}, \forall k,j \in \{1,\dots,K\},\\
  &q_k \ge 0, \pi_k \ge 0, \theta_k \ge 0, \forall k \in \{1,\dots,K\}.
\end{aligned}
\end{equation}
The first two constraints correspond to IR and IC, respectively. Note that the EAP will reveal its private type truthfully with the IR and IC constraints. Specifically, the IR condition ensures the EAP's participation and the IC condition ensures that each EAP selects the contract item designed for its corresponding type to gain highest payoff.

\subsubsection{Constraint Reduction}
There are $K$ IR constraints and $K(K-1)$ IC constraints in (\ref{eq:CntFrm}), which are non-convex and couple different EAPs together. It is hard to solve (\ref{eq:CntFrm}) directly due to the complicated constraints. Motivated by this, in the subsection we first reduce the constraints of (\ref{eq:CntFrm}) and transform it.

We first realize that the following necessary conditions can be derived from the IR and IC constraints.

\begin{lemma}\label{lemma:pi2q}
For any feasible contract, $\pi_i > \pi_j$ if and only if $q_i > q_j$, $\forall i,j \in \{1,\dots,K\}$.
\end{lemma}

\begin{proof}
See Appendix \ref{appendix:lemma_pi2q}.
\end{proof}

Lemma \ref{lemma:pi2q} shows that the EAP contributing more received power at the information source will receive more reward.

\begin{lemma}\label{lemma:pi2q2}
For any feasible contract, $\pi_i = \pi_j$ if and only if $q_i = q_j$, $\forall i,j \in \{1,\dots,K\}$.
\end{lemma}

Lemma \ref{lemma:pi2q2} can be proved by using similar procedures as Lemma \ref{lemma:pi2q}, which is omitted for brevity. Lemma \ref{lemma:pi2q2} indicates that the EAPs providing the same received power will get the same amount of reward.

\begin{lemma}\label{lemma:theta2pi}
For any feasible contract, if $\theta_i > \theta_j$, then $\pi_i > \pi_j$, $\forall i,j \in \{1,\dots,K\}$.
\end{lemma}

\begin{proof}
See Appendix \ref{appendix:theta2pi}.
\end{proof}

Lemma 3 shows that a higher type EAP should be given more reward. Together with Lemma 1 and Lemma 2, it can be duduced that a higher type EAP also contributes more energy to the information source. We define this feature as monotonicity.

\begin{definition}[\textbf{Monotonicity}]
If $\theta_i \ge \theta_j, \forall i,j \in \{1,\dots,K\}$ and then $\pi_i \ge \pi_j$.
\end{definition}

Based on the above analysis, we can now use the IC condition to reduce the IR constraints and have the following lemma.

\begin{lemma}\label{lemma:ReduceIR}
With the IC condition, the IR constraints can be reduced as
\begin{equation}\label{eq:ReIR}
  \pi_1 - \frac{q_1^2}{\theta_1} \ge 0.
\end{equation}
\end{lemma}

\begin{proof}
See Appendix \ref{appendix:lemma:ReduceIR}.
\end{proof}

We can also reduce the IC constraints and attain the following lemma.

\begin{lemma}\label{lemma:ReduceIC}
With monotonicity, the IC condition can be reduced as the
local downward incentive compatibility (LDIC), given by
\begin{equation}\label{eq:LDIC}
\pi_i - \frac{q_i^2}{\theta_i} \ge \pi_{i-1} - \frac{q_{i-1}^2}{\theta_{i}}, \forall i \in \{2,\dots,K\},
\end{equation}
and the local upward incentive compatibility (LUIC), given by
\begin{equation}\label{eq:LUIC}
\pi_i - \frac{q_i^2}{\theta_i} \ge \pi_{i+1} - \frac{q_{i+1}^2}{\theta_{i}}, \forall i \in \{1,\dots,K-1\},
\end{equation}
\end{lemma}

\begin{proof}
See Appendix \ref{appendix:lemma:ReduceIC}.
\end{proof}

By using the reduced IR and IC constraints, the optimization problem (\ref{eq:CntFrm}) can be transformed as
\begin{equation}\label{eq:CntFrmRe}
\begin{aligned}
  \mathbf{(P4.5):} \qquad & \max\limits_{(\pmb \pi,\pmb q)}  \mathbb{E}\{U_{DAP}^C(\pmb \pi,\pmb q)\}  \\
  s.t. \quad & \pi_1 - \frac{q_1^2}{2 \theta_1} \ge 0, \\
  &\pi_i - \frac{q_i^2}{\theta_i} \ge \pi_{i-1} - \frac{q_{i-1}^2}{\theta_i}, \forall i \in \{2,\dots,K\},\\
  &\pi_i - \frac{q_i^2}{\theta_i} \ge \pi_{i+1} - \frac{q_{i+1}^2}{\theta_i}, \forall i \in \{1,\dots,K-1\},\\
  &\pi_K \ge \pi_{K-1} \ge \dots \ge \pi_1, \\
  &q_k \ge 0, \pi_k \ge 0, \theta_k \ge 0, \forall k \in \{1,\dots,K\}.
\end{aligned}
\end{equation}

The LDIC and the LUIC in (\ref{eq:CntFrmRe}) can be combined as shown in Lemma 8.

\begin{lemma}\label{lemma:final}
Since the optimization objective function is an increasing function of $q_k$ and a decreasing function of $\pi_k,\forall k \in \{1,\dots,K\}$, the above optimal problem can be further simplified as
\begin{equation}\label{eq:CntFrmRe2}
\begin{aligned}
  \mathbf{(P4.6):} \qquad &\max\limits_{(\pmb \pi,\pmb q)}  \mathbb{E}\{U_{DAP}^C(\pmb \pi,\pmb q)\}  \\
  s.t. \quad & \pi_1 - \frac{q_1^2}{\theta_1} = 0, \\
  &\pi_k - \frac{q_k^2}{\theta_k} = \pi_{k-1} - \frac{q_{k-1}^2}{\theta_k}, \forall k \in \{2,\dots,K\},\\
  &\pi_K \ge \pi_{K-1} \ge \dots \ge \pi_1, \\
  &q_k \ge 0, \pi_k \ge 0, \theta_k \ge 0, \forall k \in \{1,\dots,K\}.
\end{aligned}
\end{equation}
\end{lemma}

\begin{proof}
See Appendix \ref{appendix:lemma:final}.
\end{proof}

\subsubsection{Solution to Optimal Contract}
We now solve the optimization problem (\ref{eq:CntFrmRe2}) to attain the optimal contract in the subsequent way: a standard method is first applied to resolve the relaxed problem without monotonicity and the solution is then verified to satisfy the condition of the monotonicity. By iterating the first and second constraints in (\ref{eq:CntFrmRe2}), we have
\begin{equation}\label{eq:iter}
\begin{aligned}
\pi_k &= \frac{q_1^2}{\theta_1} + \sum\limits_{n=2}^k \frac{q_n^2 - q_{n-1}^2}{\theta_n}\\
&=\frac{1}{\theta_k}q_{k}^2 + \sum\limits_{n=2}^k \left(\frac{1}{\theta_{n-1}} - \frac{1}{\theta_{n}} \right)q_{n-1}^2,
\end{aligned}
\end{equation}
where $\forall k \in \{2,\dots,K\}$.
Substitute (\ref{eq:iter}) into $\mathbb{E}\{U_{DAP}^C(\pmb \pi,\pmb q)\}$, and all $\pi_k, \forall k \in \{1,\dots,K\}$ are removed from the optimization problem (\ref{eq:CntFrmRe2}), which becomes
\begin{equation}\label{eq:finalprob}
\begin{aligned}
  &\max\limits_{\pmb q} \sum\limits_{n_1=0}^{N}\sum\limits_{n_2=0}^{N-n_1}\dots \sum\limits_{n_{K-1}=0}^{N-\sum\nolimits_{i=0}^{K-2}n_i} \Phi_{n_1,\dots,n_K}\\
  &\times \left[ W \log_2 \left( 1 + \gamma \sum\limits_{k=1}^{K} n_k q_k \right) \right.\\
  &- \left. \sum\limits_{k=1}^{K-1}
\left( \frac{1}{\theta_{k}} \sum\limits_{i=k}^{K}n_i
- \frac{1}{\theta_{k+1}} \sum\limits_{i=k+1}^{K}n_i \right)q_k^2
-\frac{n_K}{\theta_K}q_K^2 \right],\\
s.t. &~~~~q_k \ge 0, \forall k \in \{1,\dots,K\}.
\end{aligned}
\end{equation}
Note that (\ref{eq:finalprob}) is composed of logarithmic functions and quadratic functions, both of which are concave functions. And the positive summation of all these concave function is still a concave function. Besides, the constraint set is a convex set. So we can leverage standard convex optimization tools in \cite{Boyd2004Convex} to solve it to get $q_k$, and then $\pi_k$ can be calculated by (\ref{eq:iter}). Moreover, monotonicity is met automatically when the type is uniformly distributed\cite{bolton2005contract}. So far, we have derived the optimal contract $(q_k,\pi_k)$, $\forall k \in \{1,\dots,K\}$, which can maximize the utility of the DAP and satisfy the constraints of IR and IC.

\subsection{Practical Implementation}
To implement the proposed approach in a practical radio frequency energy harvesting-based IoT system. The following steps should be followed.

First, the DAP needs to collect the information it requires by the computation of the optimal contract. The active sensor will broadcast pilots to allow the DAP and EAPs to estimate the channels such that the DAP is aware of the channel gain from the DAP to the sensor. From historical data, the DAP can obtain empirical values of the energy harvesting efficiency factor and noise power, and thus it can attain the value of parameter $\gamma$. With the known values of other public system parameters including the channel bandwidth, the user number, and the set of EAP types, the DAP can calculate the optimal contract.

Next, the DAP will broadcast the optimal contract to the candidate EAPs via the corresponding backhauls. By evaluating the contract, the EAPs will decide whether to participate in the cooperation. If it decides to participate in the current energy trading, it will send a feedback to the DAP. After the DAP received the feedback, it will sign a contract with the EAP.

Finally, after the contracts are signed, the EAPs will perform the contracts by establishing an energy transfer link towards the active sensor and charge it according to the agreed transmit power. When the DAP detects that the EAPs have fulfilled its contractual obligation, the DAP will pay the EAP with agreed amount of rewards via the backhaul connecting the operators.

\section{Benchmark Schemes with Complete Information}
To investigate the impacts from information scenarios and compare the proposed schemes with the existing schemes under complete information, we first extend existing Stackelberg game from unified pricing strategy into discriminative pricing strategy. And then we present the centralized optimization scheme under complete information as the reference for the proposed incentive mechanisms.

\subsection{Stackelberg Game Formulation}
To fully exploit the potentials of EAPs with distinct channel conditions and energy costs, a discriminative pricing strategy is considered, i.e., the DAP can impose different prices of per unit energy harvested from different EAPs. The utility function of the DAP can be rewritten as
\begin{equation}
U_{DAP}^S(\pmb \lambda,\pmb q)=R_{SD} -  \sum\limits_{m=1}^{N} \lambda_m q_m.
\end{equation}
where $\pmb q = [q_1,q_2,\dots,q_N]^T$ is the vector of the active sensor's received power from EAPs, with $q_m$ denoting the received power from the $m$th EAP, $\pmb \lambda = [\lambda_1,\lambda_2,\dots,\lambda_N]^T$ is the vector of prices per unit energy harvested from EAPs, with $\lambda_m \ge 0$ denoting the price per unit energy harvested from $m$ EAP, and $R_{SD}$ is the achievable throughput defined in (\ref{eq:Ras}) and (\ref{eq:Ras2}). The optimization problem for the DAP or the leader-level game can be formulated as
\begin{equation}\label{eq:opt_DAP}
\begin{aligned}
  \mathbf{(P5.1):} \quad & \max\limits_{\pmb \lambda} U_{DAP}^S(\pmb \lambda,\pmb q) \\
  s.t. \quad & \pmb \lambda \ge \pmb 0
\end{aligned}
\end{equation}
Note that the optimization problem (P5.1) is different from (P4.1) under asymmetric information, the instantaneous utility of the DAP is optimized here, instead of expected utility of the DAP in asymmetric information.

Each EAP is modeled as a follower which would like to maximize its individual profit, the utility of which is rewritten as
\begin{equation}\label{eq:USEAP}
U_m^S(\lambda_m,q_m) = \lambda_m q_m - \frac{a_m}{G_{m,s}^2}q_m^2,
\end{equation}
where $a_m > 0$ and $G_{m,s}$ are the energy cost coefficient and channel gain of the $m$th EAP, respectively. Thus, the optimization problem for the EAP $m$ or the follower-level game is given by
\begin{equation}\label{eq:opt_EAP}
\begin{aligned}
  \mathbf{(P5.2):} \quad & \max\limits_{q_m} U_m^S(\lambda_m,q_m) \\
  s.t. \quad & q_m \ge 0 \quad\quad\quad\quad\quad\quad
\end{aligned}
\end{equation}

\subsection{Analysis of the Formulated Stackelberg Game}
In this subsection, we will derive the SE of the formulated game by analyzing the optimal strategies for the DAP and EAPs to maximize their own utility functions. A closed-form solution is derived by using Karush-Kuhn-Tucker (KKT) conditions.

First, the optimal $q_m^*$ of the $m$th EAP is similar to that in (\ref{eq:EAP_optimal_solution_asym}), which is given by the following Lemma:
\begin{lemma}\label{lemma:EAP_optimal_solution}
For given $\lambda_m$, the optimal solution for problem (P5.2) is given by
\begin{align}\label{eq:EAP_optimal_solution}
{q_m^{{*}}} = \frac{G_{m,s}^2 \lambda_m}{2 a_m}.
\end{align}
\end{lemma}
\begin{proof}
The proof of this lemma follows by noting that the objective function of problem (P5.2) given in (\ref{eq:opt_EAP}) is a concave function in terms of $q_m$.
\end{proof}
It can be observed from Lemma \ref{lemma:EAP_optimal_solution} that for the same energy price, an EAP with better channel gain and/or less energy cost would like to contribute more power to the active sensor.

Subsequently, we need to solve problem (P5.1) by replacing $q_m$ with $q_m^*$ given in (\ref{eq:EAP_optimal_solution}). The optimization problem at the DAP side can be expressed as
\begin{equation}\label{eq:opt_DAP_uniform}
\begin{aligned}
  \mathbf{(P5.3):} \quad & \max\limits_{\pmb \lambda} U_{DAP}^{S}(\pmb \lambda,\pmb q^*) \\
  s.t. \quad & \pmb \lambda \ge \pmb 0
\end{aligned}
\end{equation}
where $U_{DAP}^{S}(\pmb \lambda,\pmb q^*)$ is given by
\begin{equation}\label{eq:USEAP_Bar}
\begin{aligned}
U_{DAP}^{S}(\pmb \lambda,\pmb q^*) &= W \log_2 \left( 1 + \gamma \sum\limits_{m=1}^{N} \frac{G_{m,s}^2 }{2 a_m}\lambda_m \right) \\
&- \sum\limits_{m=1}^{N} \frac{G_{m,s}^2 }{2 a_m}\lambda_m^2.
\end{aligned}
\end{equation}

We can observe that problem (P5.3) is a concave function in terms of vector $\pmb \lambda$ since the former part in (\ref{eq:USEAP_Bar}) is a logarithm function (concave) and the latter parts in (\ref{eq:USEAP_Bar}) are the summation of quadratic functions (concave), and the constraint is affine. So problem (P5.3) is a convex optimization problem. By using KKT conditions to solve problem (P5.3), the closed-form solution for $\pmb \lambda$ is derived in the following proposition.
\begin{proposition}\label{proposition:KKT_solution}
The optimal solution to problem (P5.3) is given by
\begin{equation}\label{eq:KKT_sol}
\lambda_1^* = \lambda_2^* = \dots = \lambda_N^*=\frac{\sqrt{\log_2(e)\gamma^2 W \Theta+1}-1}{\gamma \Theta},
\end{equation}
where $e$ is the base of the natural logarithm, $\Theta$ is given by
\begin{equation}\label{eq:Theta}
\Theta = \sum\limits_{m=1}^{N} \frac{G_{m,s}^2 }{a_m}.
\end{equation}
\end{proposition}
\begin{proof}
See Appendix \ref{appendix:prop_1}.
\end{proof}
Proposition \ref{proposition:KKT_solution} shows that the optimal prices for the Stackelberg game with complete information are the same. This result is consistent with that of the Stackelberg game with asymmetric information. As we explained before, this is because the received power price is used in the Stackelberg games with complete and asymmetric information. The DAP has no motivation to treat the received power from EAPs differently.

Note that the Stackelberg game under complete information can calculate an optimal price for each instantaneous channel realization and equivalently the combination of EAPs' types. As such, it can adapt to the change of channel conditions. As a comparison, the Stackelberg game under asymmetric information can only calculate a single price no matter of the change of the channel conditions, i.e., the change of the combinations of the EAPs' types.

\subsection{Centralized Optimization}
In this part, the performance of centralized optimization scheme, i.e., the optimal contract with complete information, where the DAP knows exactly the types of the EAPs, is presented. The centralized optimization problem is given as follows.
\begin{equation}\label{eq:centralized}
\begin{aligned}
   \mathbf{(P5.4):} \qquad & \max\limits_{(\pmb \pi,\pmb q)}  \mathbb{E}\{U_{DAP}^C(\pmb \pi,\pmb q)\} \\
  s.t. \quad & \pi_k - \frac{q_k^2}{\theta_k} \ge 0, \forall k \in \{1,\dots,K\}
\end{aligned}
\end{equation}
where $\mathbb{E}\{U_{DAP}^C(\pmb \pi,\pmb q)\}$ is given in (\ref{eq:Expect_U_DAP_contract}).

Since the DAP knows exactly the types of the EAPs, the optimal prices are given by
\begin{equation}\label{eq:opt_prices_central}
\pi_k^* = \frac{q_k^2}{\theta_k}, \forall k \in \{1,\dots,K\}
\end{equation}
We substitute $\pi_k$ in (\ref{eq:centralized}) with $\pi_k^*$ and get
\begin{equation}\label{eq:centralized2}
\begin{aligned}
   \mathbf{(P5.5):} \qquad & \max\limits_{\pmb q}  \mathbb{E}\{U_{DAP}^C(\pmb \pi^*,\pmb q)\}\\
   s.t. \quad & \pmb q \ge \pmb 0
\end{aligned}
\end{equation}
where $\mathbb{E}\{U_{DAP}^C(\pmb \pi^*,\pmb q)\}$ is given by
\begin{equation}\label{eq:Expect_U_DAP_contract_central}
\begin{aligned}
  &\mathbb{E}\{U_{DAP}^C(\pmb \pi^*,\pmb q)\} = \sum\limits_{n_1=0}^{N}\sum\limits_{n_2=0}^{N-n_1}\dots \sum\limits_{n_{K-1}=0}^{N-\sum\nolimits_{i=0}^{K-2}n_i}\\
  &\left\{ \Phi_{n_1,\dots,n_K} \left[ W \log_2 \left( 1 + \gamma \sum\limits_{k=1}^{K} n_k q_k \right) - \sum\limits_{k=1}^K n_k \frac{q_k^2}{\theta_k} \right] \right\}.
\end{aligned}
\end{equation}

Note that (\ref{eq:Expect_U_DAP_contract_central}) is exactly the expectation of the social welfare, which is defined in (\ref{eq:SocWel}). Although we originally optimize the utility function of the DAP in problem (P5.4), it is consistent with the optimization of the social welfare, which is similar case in the design of contract theory as we mentioned before.

It can also be observed that problem (P5.5) is a convex optimization problem. This is because each term in the summation is composed a logarithm function (concave) and quadratic functions (concave), the summation of concave functions are still a concave function, and the constraint is affine. We can get the solution of problem (P5.5) by solving the system of equations given by KKT conditions, which is omitted here as it is similar to that in Appendix \ref{appendix:Prop_2}.

\section{Simulations and Discussions}
In this section, we first evaluate the feasibility of the proposed contract, and then compare the performance of the proposed incentive mechanisms. The performance of centralized optimization scheme is also simulated as the upper bound.

\begin{table}[!t]
\renewcommand{\arraystretch}{1.3}
\caption{System Settings}
\label{system_settings}
\centering
\begin{tabular}{|c|c|}
\hline
Parameters                                      & Values        \\
\hline
Energy harvesting efficiency $\eta$             & 0.5           \\
\hline
Bandwidth $W$                                   & 1MHz          \\
\hline
Energy cost coefficient $a_m$                   & [0.1,1]        \\
\hline
$d_{m,s}$                                       & [5m,10m]        \\
\hline
$d_{a,s}$                                       & [15m,25m]       \\
\hline
Path-loss coefficient $\alpha$                  & 2             \\
\hline
Power attenuation at reference distance of 1m   & 30dB         \\
\hline
Noise power $N_0$                               & $10^{-8}mW$   \\
\hline
\end{tabular}
\end{table}

\begin{figure}[!t]
  \centering\scalebox{0.5}{\includegraphics{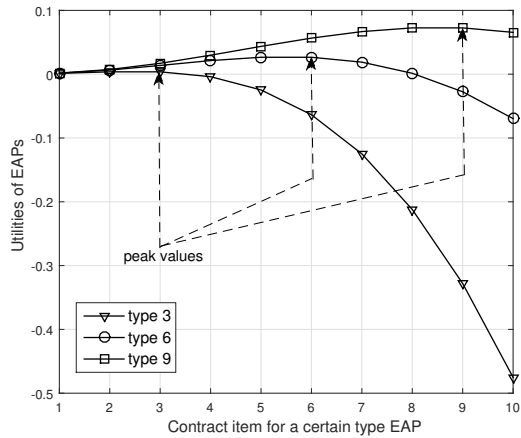}}
  \caption{Utilities of EAPs with type 3, type 6 and type 9 as functions of contract items designed for all kinds of EAPs from type 1 to type 10. We set $N=5$ and $K=10$.}\label{Fig:IC}
\end{figure}

The main system parameters are shown in Table I. Since $\theta=G_{m,s}^2/a_m$ and $\gamma=\eta G_{a,s}/N_0$, the practical ranges of $\theta$ and $\gamma$ can be determined by the parameters shown in Table I. In the simulations, K types of EAPs are first generated randomly and used as the set of EAP types. Then each of N EAPs in the market will choose one type from the set of EAP types uniformly, and thus the DAP's type $\theta$ is uniformly distributed. The unit of achievable throughput is set as Mbps.

To verify the feasibility (i.e., IR and IC) of the proposed scheme under information asymmetry, the utilities of EAPs with type 3, type 6 and type 9 are plotted in Fig. \ref{Fig:IC} as functions of all contract items $(q_k,\pi_k),k\in {1,2,\dots,K}$. We can see from Fig. \ref{Fig:IC} that each of the utility achieves its peak value only when it chooses the contract item designed for its corresponding type, which indicates that the IC constraint is satisfied. For example, for the type 6 EAP, its utility achieves the peak value only when it selects the contract item $(q_6,\pi_6)$, which is exactly designed for its type. If the type 6 EAP selects any other contract item $(q_k,\pi_k),k\in {1,2,\dots,K}$ and $k\ne 6$, its utility will reduce. Moreover, when each of above type EAPs (i.e., type 3, type 6 and type 9) chooses the contract item designed for its corresponding type, the utilities are nonnegative. Note that similar phenomenon can be observed for all other types of EAPs when they select the contract item designed for their corresponding types, which are not shown in Fig. \ref{Fig:IC} for brevity. In this sense, the IR condition is satisfied. It can be concluded that utilizing the proposed scheme, EAPs will automatically reveal its type to the DAP after selecting the contract item. This means that using the proposed scheme, the DAP can capture the EAPs' private information (i.e., its type), and thus effectively address the problem of information asymmetry.

\begin{figure}[!t]
  \centering\scalebox{0.5}{\includegraphics{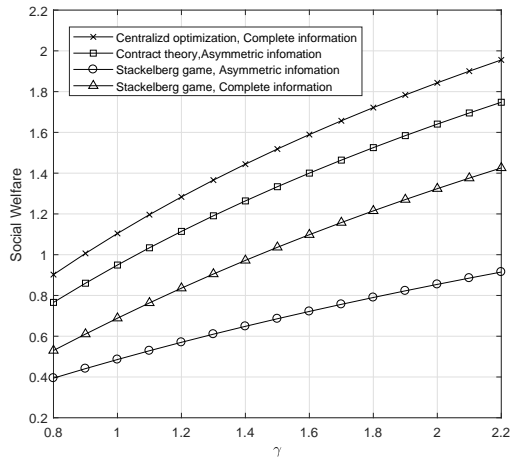}}
  \caption{Social welfare as a function of $\gamma$. We set $N=2$ and $K=5$.}\label{Fig:Social}
\end{figure}

\begin{figure}[!t]
  \centering\scalebox{0.5}{\includegraphics{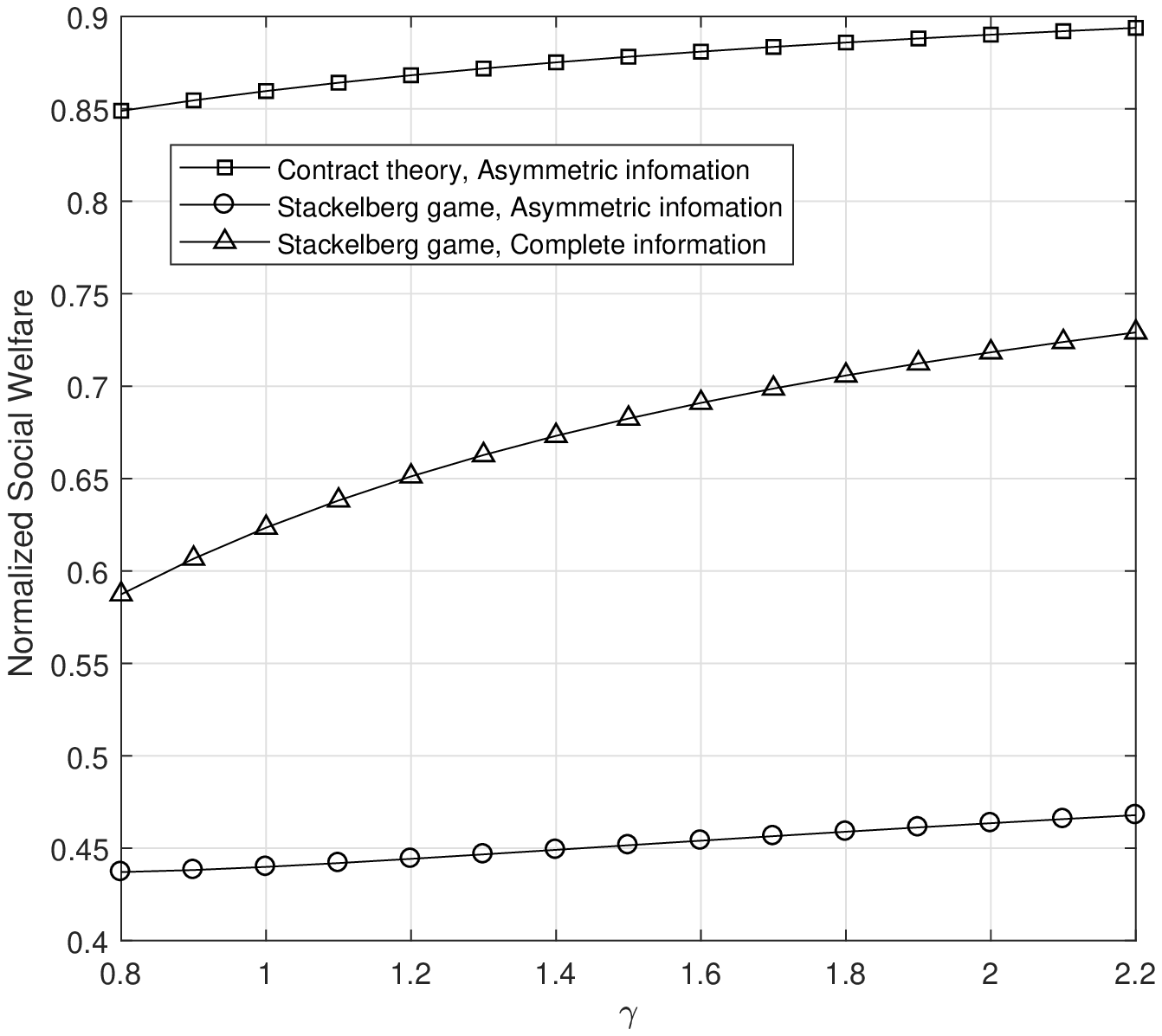}}
  \caption{Normalized social welfare as a function of $\gamma$. We set $N=2$ and $K=5$.}\label{Fig:SocialNormalized}
\end{figure}

\begin{figure}[!t]
  \centering\scalebox{0.48}{\includegraphics{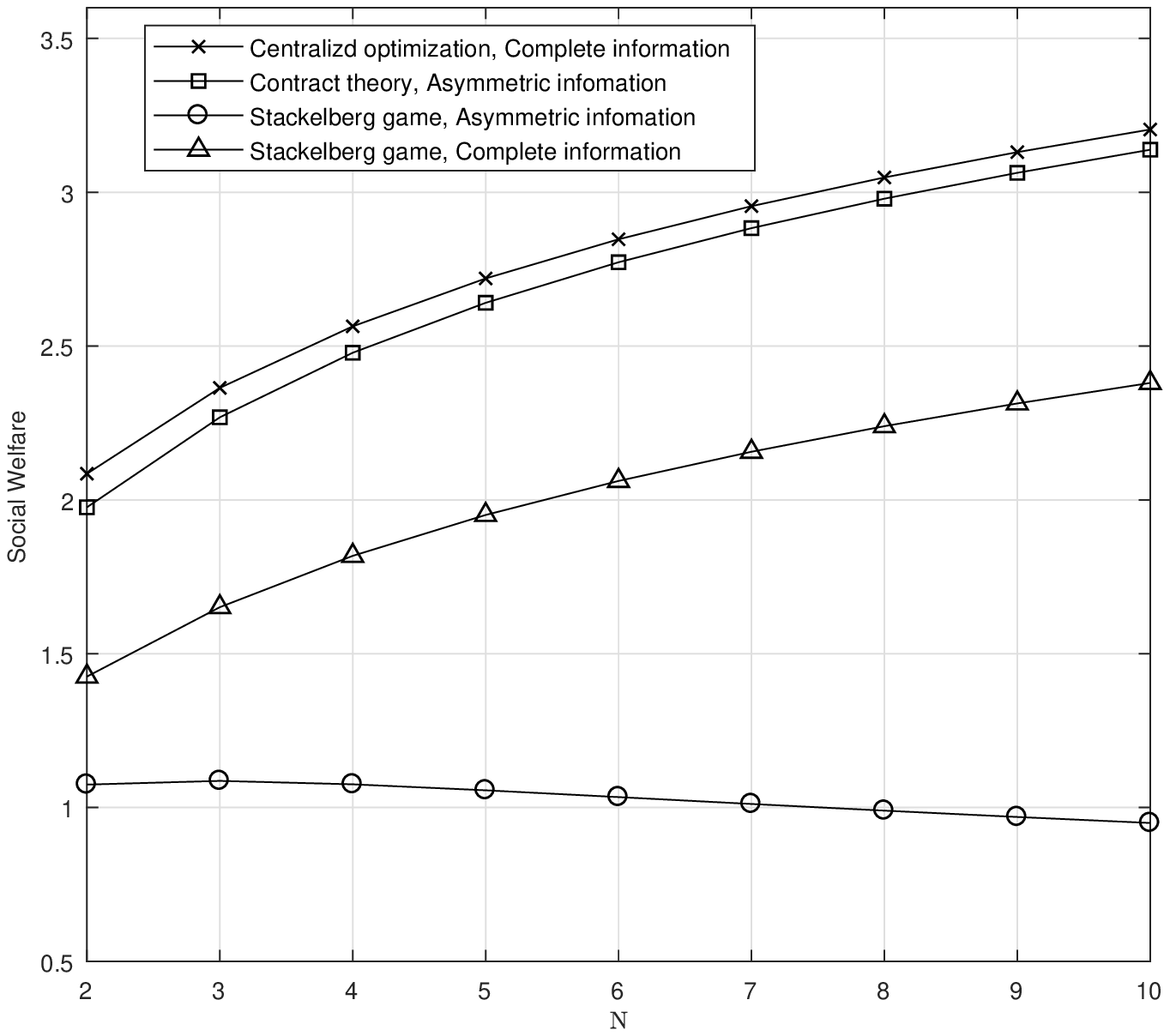}}
  \caption{Social welfare as a function of $N$. We set $K=2$, $\gamma=2.2$ and $K=2,3,\dots,10$.}\label{Fig:Social_N}
\end{figure}

\begin{figure}[!t]
  \centering\scalebox{0.5}{\includegraphics{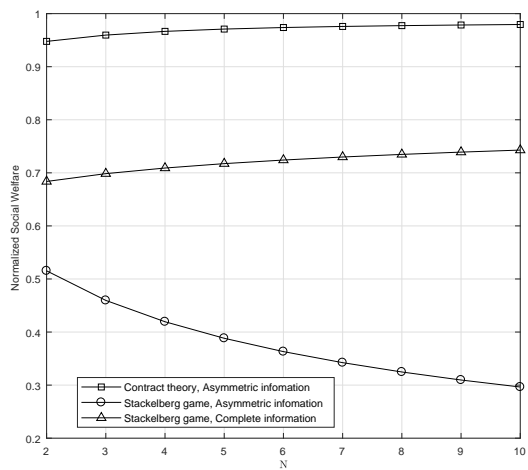}}
  \caption{Normalized social welfare as a function of $N$. We set $K=2$, $\gamma=2.2$ and $N=2,3,\dots,10$.}\label{Fig:SocialNormalized_N}
\end{figure}


To evaluate the performance of the proposed schemes, we compare the social welfare of the contract, Stackelberg games and the upper bound. Fig. \ref{Fig:Social} plots the social welfare of these schemes as a function of $\gamma$. It can be observed from Fig. \ref{Fig:Social} that the utilities achieved by all schemes increase with $\gamma$. This is because with the same $\sum\nolimits_{m=1}^{N} q_m$, the larger the value of $\gamma$, the larger the achievable throughput $R_{sa}$ (refer to (\ref{eq:Ras2})), and thus larger social welfare (refer to (\ref{eq:SocWel})). The performance of the optimal scheme with complete information providing the best performance serving as the upper bound. The performance of contract scheme is generally better than that of two Stackelberg games. This is because in contract theory, the EAPs have limited contract items to choose from and thus by using the contract theory, the DAP extracts more benefits from the EAPs and leave less surplus for the EAPs. However, in Stackelberg games, the EAPs have the freedom to optimize its individual utility function and thus can reserve more surplus. So the performance of the Stackelberg games are inferior than that of the contract scheme. We can also observe that the Stackelberg game with asymmetric information is inferior than that with complete information. This is because without complete information, the Stackelberg game fails to adapt to the change of the channel, and thus the performance becomes worse.

Fig. \ref{Fig:SocialNormalized} shows the normalized social welfare as a function of $\gamma$, where social welfare of the contract and Stackelberg games are normalized by the upper bound. It can be seen in Fig. \ref{Fig:SocialNormalized} that when $\gamma$ is small, the social welfare of contract can initially achieve more than $85\%$ of that of the centralized optimization scheme with complete information, and gradually approach to it with the increasing of $\gamma$. This demonstrates that the proposed incentive mechanism can effectively mitigate the effects of information asymmetry by leveraging contract theory. While the performance of the Stackelberg game with complete information is generally less than $75\%$ of that of the optimal scheme with complete information. Moreover, the performance of the Stackelberg game with asymmetric information is even worse, which is generally less than $50\%$ of that of the optimal scheme with complete information. The above results show that by using the monopoly position in contract theory to provide limited contract items, the contract can achieve good performance close to the optimal centralized optimization with complete information. However, in Stackelberg games, the DAP grants some freedom for the EAPs to do optimizations, which are selfish and do not care about social welfare. As such, its performance in terms of social welfare is degraded.

To explore the impact of total EAP number $N$ in the market, we plot the curves of the social welfare and the normalized social welfare of the contract, Stackelberg games and the upper bound in Fig. \ref{Fig:Social_N} and Fig. \ref{Fig:SocialNormalized_N}. In Fig. \ref{Fig:Social_N}, the social welfare of these schemes is plotted as a function of $N$. We can observe from Fig. \ref{Fig:Social_N} that the utility functions achieved by all three schemes of upper bound, contract, Stackelberg with complete information increase with $N$. This is because the overall social welfare increases with the number of EAPs in the market. The more EAPs in the market, the larger the summation of utility functions of all the DAP and EAPs. However, the Stackelberg game under asymmetric information decrease slightly as $N$ increases. This is because the Stackelberg game under asymmetric information fails to adapt to the change of the combinations of the EAPs' types. As we mentioned before, it can only calculate one single price for all the combinations of EAPs' types. The more EAPs in the market, the more diverse combinations of the EAPs' types. As such the performance of the Stackelberg under asymmetric information become worse. As a comparison, the Stackelberg game with complete information can calculate a price targeting a certain combination of EAPs' types in the market. So it provides better performance than that of its asymmetric counterpart. While the contract leverages its monopoly status in the market structure to provide a limited group of contract items for the EAPs to choose from. Therefore, contract theory-based scheme provides the better performance than that of both Stackelberg games and close to the performance of the upper bound.

In Fig. \ref{Fig:SocialNormalized_N}, the normalized social welfare of these schemes is plotted as a function of $N$, where social welfare of the contract and Stackelberg games are normalized by the upper bound. It can be seen in Fig. \ref{Fig:SocialNormalized_N} that when $N=2$, the social welfare of contract can initially gain more than $95\%$ of that of the centralized optimization scheme with complete information, and gradually approach to it with the increasing of $N$. This proves that the effects of information asymmetry can be mitigated successfully by leveraging contract theory. While the Stackelberg game with complete information can only provide the normalized social welfare of less than $75\%$. Besides, the performance of the Stackelberg game with asymmetric information is even worse, which is generally less than $50\%$ of that of the optimal scheme with complete information and decrease significantly with the increasing of $N$ in the market. This is because the more EAPs in the market, the more diverse the combinations of EAPs' types will be. The Stackelberg game under asymmetric information cannot adapt to the change of the combinations of EAPs' types as it can only calculate a single price for all possible combinations of EAPs' types.

\section{Conclusions and Future Work}
In this paper, we developed incentive mechanisms under complete and asymmetric information to unveil the impact of information asymmetry and market structure. Specifically, we developed a contract based incentive mechanisms for the wireless energy trading in radio frequency energy harvesting (RFEH) based Internet of Things (IoT) systems under asymmetric information. In the asymmetric information scenario, a Stackelberg game based scheme is also formulated as a comparison. In complete information, the existing Stackelberg game is extended from unified pricing into discriminative pricing as a comparison. In the simulations, it was shown that the Stackelberg game degrades significantly without complete information, and the performance of the contract scheme under asymmetric information is better than that of the Stackelberg scheme with complete information. It can be concluded that the performance of the considered system depends largely on the market structure (i.e., whether the EAPs are allowed to optimize their received power at the IoT devices with full freedom or not) than on the information scenarios (i.e., the complete or asymmetric information).

In our future work, we could consider both information asymmetry as well as hidden action. In this scenario, the DAP is not aware of the private information of EAPs and it cannot distinguish the actions taken by different EAPs, i.e., the received power contributed by different EAPs. Because the actions of EAPs are hidden from the DAP, some EAPs may get the reward of the group without paying any efforts, which leads to the free-rider problem. In this case, another mathematical tool from the economics, known as the moral hazard in teams, has a good potential to design effective incentive mechanisms for this new scenario.

\begin{appendix}
\subsection{Proof of Proposition \ref{proposition:KKT_solution_insight}} \label{appendix:Prop_2}
In this part, we will prove the proposition \ref{proposition:KKT_solution_insight}. Because the problem (P4.3) is a convex optimization problem, KKT conditions are the sufficient and necessary conditions for the optimal solution. The KKT conditions of problem (P4.3) are given as follows.

The first-order necessary condition are given by
\begin{equation}\label{eq:KKT1}
\left\{
\begin{aligned}
\sum\limits_{n_1,\dots,n_K} \Phi_{n_1,\dots,n_K} &\left[\frac{-W\log_2(e) \gamma n_1 \theta_1}{2+\gamma \sum\limits_{k=1}^{K}n_k \theta_k \lambda_k} + n_1 \theta_1 \lambda_1 \right] \\
&- \mu_1 = 0\\
\sum\limits_{n_1,\dots,n_K} \Phi_{n_1,\dots,n_K} &\left[\frac{-W\log_2(e) \gamma n_2 \theta_2}{2+\gamma \sum\limits_{k=1}^{K}n_k \theta_k \lambda_k} + n_2 \theta_2 \lambda_2 \right] \\
&- \mu_2 = 0\\
&\qquad\qquad\qquad\qquad\qquad\qquad \vdots\\
\sum\limits_{n_1,\dots,n_K} \Phi_{n_1,\dots,n_K} &\left[\frac{-W\log_2(e) \gamma n_N \theta_N}{2+\gamma \sum\limits_{k=1}^{K}n_k \theta_k \lambda_k} + n_N \theta_N \lambda_N \right] \\
&- \mu_N = 0
\end{aligned}
\right.
\end{equation}
where $\theta_k \ge 0$ are the types of EAPs, $\mu_k \ge 0$ are KKT multipliers, $\lambda_m \ge 0$ are the prices, and $k=1,2,\dots,K$.
The complementary slackness condition is given by
\begin{equation}\label{eq:KKT2}
\mu_1 \lambda_1 + \mu_2 \lambda_2 + \dots + \mu_K \lambda_K = 0.
\end{equation}
Since $\mu_k \ge 0$ and $\lambda_k \ge 0$, $k=1,2,\dots,K$ hold, (\ref{eq:KKT2}) becomes
\begin{equation}\label{eq:KKT3}
\left\{
\begin{aligned}
&\mu_1 \lambda_1  = 0\\
&\mu_2 \lambda_2 = 0\\
&\qquad \vdots\\
&\mu_K \lambda_K = 0
\end{aligned}
\right.
\end{equation}
To get the optimal solution of the KKT conditions, we need to solve the equation system consists of (\ref{eq:KKT1}) and (\ref{eq:KKT3}) in terms of $\mu_k \ge 0$ and $\lambda_k \ge 0$, $k=1,2,\dots,N$, which is a system of quadratic equations.
Now we will discuss the combinations of active or inactive constraints in KKT condition. Let first test if $\mu_1=\mu_2=\dots=\mu_K=0$ leads to a valid solution. By substitute $\mu_1,\mu_2,\dots,\mu_K$ with $\mu_1=\mu_2=\dots=\mu_K=0$ in (\ref{eq:KKT1}), we have
\begin{equation}\label{eq:KKT4}
\left\{
\begin{aligned}
&\lambda_1 \Delta = W\log_2(e) \gamma \Omega(\pmb \lambda)\\
&\lambda_2 \Delta = W\log_2(e) \gamma \Omega(\pmb \lambda)\\
&\qquad\qquad\qquad \vdots\\
&\lambda_K \Delta = W\log_2(e) \gamma \Omega(\pmb \lambda)\\
\end{aligned}
\right.
\end{equation}
where $\Delta$ is given by
\begin{equation}\label{eq:Delta}
\Delta = \sum\limits_{n_1,\dots,n_K} \Phi_{n_1,\dots,n_K} = 1,
\end{equation}
and $\Omega(\pmb \lambda)$, $\pmb \lambda = \left[ \lambda_1,\lambda_2,\dots,\lambda_K \right]^T$  is given by
\begin{equation}\label{eq:Omega}
\Omega(\pmb \lambda) = \sum\limits_{n_1,\dots,n_K} \frac{\Phi_{n_1,\dots,n_K}}{2+\gamma \sum\limits_{k=1}^{K}n_k \theta_k \lambda_k}.
\end{equation}

The above system of equations in (\ref{eq:KKT4}) can be solved numerically. Note that the right term of each equation in (\ref{eq:KKT4}) are the same, so we can conclude that
\begin{equation}\label{eq:KKT5}
\lambda_1 = \lambda_2 =\dots = \lambda_K = W\log_2(e) \gamma \Omega(\pmb \lambda).\\
\end{equation}

Because problem (P4.3) is a convex optimization problem, we can conclude that this solution given by KKT conditions is the solution of original optimization problem.

\subsection{Proof of Lemma \ref{lemma:pi2q}}\label{appendix:lemma_pi2q}
The proof is conducted in two parts. First, we prove if $q_i > q_j$, then $\pi_i > \pi_j$. Due to the IC constraints in (\ref{eq:CntFrm}), we have
\begin{equation}
\pi_i - \frac{q_i^2}{\theta_i} \ge \pi_j - \frac{q_j^2}{\theta_i},
\end{equation}
and equivalently,
\begin{equation}
\theta_i (\pi_i - \pi_j) \ge q_i^2 - q_j^2 = (q_i + q_j)(q_i - q_j).
\end{equation}

Since $q_i > q_j$, we have
\begin{equation}
\theta_i (\pi_i - \pi_j) \ge q_i^2 - q_j^2 = (q_i + q_j)(q_i - q_j) >0,
\end{equation}
and thus $\pi_i > \pi_j$.

Next we prove if $\pi_i > \pi_j$, then $q_i > q_j$. Due to the IC constraints in (\ref{eq:CntFrm}), we have
\begin{equation}
\pi_j - \frac{q_j^2}{\theta_j} \ge \pi_i - \frac{q_i^2}{\theta_j},
\end{equation}
and equivalently,
\begin{equation}
(q_i + q_j)(q_i - q_j) = q_i^2 - q_j^2 \ge \theta_i (\pi_i - \pi_j).
\end{equation}

Since $\pi_i >\pi_j$, then we have
\begin{equation}
(q_i + q_j)(q_i - q_j) = q_i^2 - q_j^2 \ge \theta_i (\pi_i - \pi_j) > 0,
\end{equation}
and thus $q_i > q_j$. This completes the proof.

\subsection{Proof of Lemma \ref{lemma:theta2pi}}\label{appendix:theta2pi}
We prove this by contradiction. Suppose that there exists $\pi_i < \pi_j$ when $\theta_i > \theta_j$. We have
\begin{equation}\label{eq:Lemma3-1}
(\pi_i - \pi_j)(\theta_i - \theta_j) < 0.
\end{equation}
Due to the IC constraints, we also have
\begin{equation}\label{eq:Lemma3-2}
\pi_i - \frac{q_i^2}{\theta_i} \ge \pi_j - \frac{q_j^2}{\theta_i},
\end{equation}
and
\begin{equation}\label{eq:Lemma3-3}
\pi_j - \frac{q_j^2}{\theta_j} \ge \pi_i - \frac{q_i^2}{\theta_j}.
\end{equation}
Combine (\ref{eq:Lemma3-2}) and (\ref{eq:Lemma3-3}), we have
\begin{equation}\label{eq:Lemma3-4}
(\pi_i - \pi_j)(\theta_i - \theta_j) \ge 0,
\end{equation}
which is in contradiction with (\ref{eq:Lemma3-1}). So if $\theta_i > \theta_j$, then $\pi_i > \pi_j$.

\subsection{Proof of Lemma \ref{lemma:ReduceIR}}\label{appendix:lemma:ReduceIR}
Due to the IC condition, $\forall k \in \{2,\dots,K\}$, we have
\begin{equation}\label{eq:ReIRP-1}
\pi_k - \frac{q_k^2}{\theta_k} \ge \pi_1 - \frac{q_1^2}{\theta_k}.
\end{equation}
Since we have defined that $\theta_1 < \theta_2 < \dots < \theta_K$, we also have
\begin{equation}\label{eq:ReIRP-2}
\pi_1 - \frac{q_1^2}{\theta_k} \ge \pi_1 - \frac{q_1^2}{\theta_1}.
\end{equation}
Combine (\ref{eq:ReIRP-1}) and (\ref{eq:ReIRP-2}), we have
\begin{equation}\label{eq:ReIRP-3}
\pi_k - \frac{q_k^2}{\theta_k} \ge \pi_1 - \frac{q_1^2}{\theta_1} \ge 0.
\end{equation}

Note that (\ref{eq:ReIRP-3}) shows that with the IC condition, if the IR condition of the EAP with type $\theta_1$ holds, the IR condition of the other $K-1$ types will also hold. So the other $K-1$ IR conditions can be bind into the IR condition of the EAP with type $\theta_1$.

\subsection{Proof of Lemma \ref{lemma:ReduceIC}}\label{appendix:lemma:ReduceIC}
There are $K(K-1)$ IC constraints in (\ref{eq:CntFrm}), which can be divided into $K(K-1)/2$ downward incentive compatibility (DIC)\footnote{Note that $K(K-1)/2$ is still an integer. Because $K(K-1)$ is the multiplication of two continuous integers, which must be an even number. So it is divisible by two.}, given by
\begin{equation}\label{eq:DIC}
\pi_i - \frac{q_i^2}{\theta_i} \ge \pi_{j} - \frac{q_{j}^2}{\theta_{i}}, \forall i,j \in \{2,\dots,K\}, i>j,
\end{equation}
and $K(K-1)/2$ upward incentive compatibility (UIC), given by
\begin{equation}\label{eq:UIC}
\pi_i - \frac{q_i^2}{\theta_i} \ge \pi_{j} - \frac{q_{j}^2}{\theta_{i}}, \forall i,j \in \{2,\dots,K\}, i<j,
\end{equation}

Let's first prove the DIC can be reduced as the LDIC. By using the LDIC for three continuous types, $\theta_{i-1} < \theta_{i} < \theta_{i+1}, \forall i \in \{2,\dots,K-1\}$, we have
\begin{equation}\label{eq:LDICP-1}
\pi_{i+1} - \frac{q_{i+1}^2}{\theta_{i+1}} \ge \pi_{i} - \frac{q_{i}^2}{\theta_{i+1}},
\end{equation}
\begin{equation}\label{eq:LDICP-2}
\pi_i - \frac{q_i^2}{\theta_i} \ge \pi_{i-1} - \frac{q_{i-1}^2}{\theta_{i}}, \forall i.
\end{equation}
By applying the monotonicity, i.e., if $\theta_i > \theta_j$, then $\pi_i > \pi_j$, $\forall i,j \in \{1,\dots,K\}$, we have
\begin{equation}\label{eq:LDICP-3}
(\theta_{i+1} - \theta_{i})(\pi_{i} - \pi_{i-1}) \ge 0,
\end{equation}
\begin{equation}\label{eq:LDICP-4}
\theta_{i+1}(\pi_{i} - \pi_{i-1}) \ge \theta_{i}(\pi_{i} - \pi_{i-1}),
\end{equation}
Combine (\ref{eq:LDICP-2}) and (\ref{eq:LDICP-4}), we have
\begin{equation}\label{eq:LDICP-5}
\theta_{i+1}(\pi_{i} - \pi_{i-1}) \ge \theta_{i}(\pi_{i} - \pi_{i-1}) \ge q_i^2 - q_{i-1}^2.
\end{equation}
Equally, (\ref{eq:LDICP-5}) becomes
\begin{equation}\label{eq:LDICP-6}
\pi_{i} - \frac{q_i^2}{\theta_{i+1}} \ge \pi_{i-1} - \frac{q_{i-1}^2}{\theta_{i+1}}.
\end{equation}
Combine (\ref{eq:LDICP-6}) and (\ref{eq:LDICP-1}), we have
\begin{equation}\label{eq:LDICP-7}
\pi_{i+1} - \frac{q_{i+1}^2}{\theta_{i+1}} \ge \pi_{i-1} - \frac{q_{i-1}^2}{\theta_{i+1}}.
\end{equation}
So far, we have proved that type $\theta_{i+1}$ will prefer contract item $(q_{i+1},\pi_{i+1})$ rather than contract item $(q_{i-1},\pi_{i-1})$. By using (\ref{eq:LDICP-7}), it can be extended downward until type $\theta_{1}$, and thus all DIC holds.
\begin{equation}\label{eq:LDICP-8}
\pi_{i+1} - \frac{q_{i+1}^2}{\theta_{i+1}} \ge \pi_{i-1} - \frac{q_{i-1}^2}{\theta_{i+1}} \ge \dots\ \ge \pi_{1} - \frac{q_{1}^2}{\theta_{1}}, \forall i.
\end{equation}
So we conclude that with the monotonicity and the LDIC, the DIC holds. Similarly, we can prove that with the monotonicity and the LUIC, the UIC holds.

\subsection{Proof of Lemma \ref{lemma:final}}\label{appendix:lemma:final}
We will first prove that the LDIC can be simplified as $\pi_k - q_k^2 / \theta_k = \pi_{k-1} - q_{k-1}^2/\theta_k$, which together with monotonicity can ensure the LUIC hold.

For the reduced IR constraint $\pi_1 - {q_1^2}/{\theta_1} \ge 0$ in (\ref{eq:CntFrmRe}), the DAP will lower the $\pi_1$ as possible as it can to improve the optimization objective function $\mathbb{E}\{U_{DAP}\}$, until $\pi_1 - {q_1^2}/{\theta_1} = 0$.
As for the LDIC, which is $\pi_i - {q_i^2}/{\theta_i} \ge \pi_{i-1} - {q_{i-1}^2}/{\theta_i}, \forall i \in \{2,\dots,K\}$. Notice that the LDIC will still hold if both $\pi_i$ and $\pi_{i-1}$ are lowered by the same amount. To maximize the optimization objective function, the DAP will lower all $\pi-j$ as much as possible until $\pi_i - {q_i^2}/{\theta_i} = \pi_{i-1} - {q_{i-1}^2}/{\theta_i}$. Note that this process will not affect on other type's LDIC. So the LDIC can be reduced to $\pi_i - {q_i^2}/{\theta_i} = \pi_{i-1} - {q_{i-1}^2}/{\theta_i}, \forall k \in \{2,\dots,K\}$.

Next, we show that if $\pi_i - {q_i^2}/{\theta_i} = \pi_{i-1} - {q_{i-1}^2}/{\theta_i}, \forall k \in \{2,\dots,K\}$ and the monotonicity hold, the LUIC holds. Since we have $\pi_i - {q_i^2}/{\theta_i} = \pi_{i-1} - {q_{i-1}^2}/{\theta_i}, \forall k \in \{2,\dots,K\}$, and equally it becomes
\begin{equation}\label{eq:EQP-1}
\theta_{i}(\pi_{i}-\pi_{i-1}) = q_i^2 - q_{i-1}^2.
\end{equation}
Because of monotonicity, i.e., if $\theta_i \ge \theta_{i-1}$, then $\pi_i \ge \pi_{i-1}$, we further have
\begin{equation}\label{eq:EQP-2}
\theta_{i}(\pi_{i}-\pi_{i-1}) \ge \theta_{i-1}(\pi_{i}-\pi_{i-1}).
\end{equation}
Combine (\ref{eq:EQP-1}) and (\ref{eq:EQP-2}), we have
\begin{equation}\label{eq:EQP-3}
\theta_{i}(\pi_{i}-\pi_{i-1}) = q_i^2 - q_{i-1}^2 \ge \theta_{i-1}(\pi_{i}-\pi_{i-1}),
\end{equation}
and equally we have
\begin{equation}\label{eq:EQP-4}
\theta_{i-1}\pi_{i} - q_i^2 \le \theta_{i-1}\pi_{i-1} - q_{i-1}^2,
\end{equation}
\begin{equation}\label{eq:EQP-5}
\pi_{i} - \frac{q_i^2}{\theta_{i-1}} \le \pi_{i-1} - \frac{q_{i-1}^2}{\theta_{i-1}},
\end{equation}
which is exactly the LUIC condition. So the LUIC can be removed from the constraints in (\ref{eq:CntFrmRe}).

\subsection{Proof of Proposition \ref{proposition:KKT_solution}}\label{appendix:prop_1}
In this part, we will prove the proposition \ref{proposition:KKT_solution}. Since the problem (P5.3) is a convex optimization problem, KKT conditions will be the sufficient and necessary conditions for the optimal solution. To solve problem (P5.3), the KKT conditions are given as follows.

The first-order necessary condition are given by
\begin{equation}\label{eq:KKT_FO}
\left\{
\begin{aligned}
&\frac{-W\log_2(e) \gamma \theta_1}{2+\gamma \sum\limits_{m=1}^{N}\theta_m \lambda_m} + \theta_1 \lambda_1 - \mu_1 = 0\\
&\frac{-W\log_2(e) \gamma \theta_2}{2+\gamma \sum\limits_{m=1}^{N}\theta_m \lambda_m} + \theta_2 \lambda_2 - \mu_2 = 0\\
&\qquad\qquad\qquad \vdots\\
&\frac{-W\log_2(e) \gamma \theta_N}{2+\gamma \sum\limits_{m=1}^{N}\theta_m \lambda_m} + \theta_N \lambda_N - \mu_N = 0
\end{aligned}
\right.
\end{equation}
where $\theta_m = \frac{G_{m,s}^2}{a_m}$, $\mu_m \ge 0$ are KKT multipliers, and $\lambda_m \ge 0,m=1,2,\dots,N$.

The complementary slackness condition is given by
\begin{equation}\label{eq:KKT_CS}
\mu_1 \lambda_1 + \mu_2 \lambda_2 + \dots + \mu_N \lambda_N = 0.
\end{equation}
Since $\mu_m \ge 0$ and $\lambda_m \ge 0,m=1,2,\dots,N$ hold, (\ref{eq:KKT_CS}) becomes
\begin{equation}\label{eq:KKT_CS2}
\left\{
\begin{aligned}
&\mu_1 \lambda_1  = 0\\
&\mu_2 \lambda_2 = 0\\
&\qquad \vdots\\
&\mu_N \lambda_N = 0
\end{aligned}
\right.
\end{equation}

To get the optimal solution of the KKT conditions, we need to solve the equation system consists of (\ref{eq:KKT_FO}) and (\ref{eq:KKT_CS2}) in terms of $\mu_m \ge 0$ and $\lambda_m \ge 0,m=1,2,\dots,N$, which is a system of quadratic equations.

Now we will discuss the combinations of active or inactive constraints in KKT condition. Let first test if $\mu_1=\mu_2=\dots=\mu_N=0$ leads to a valid solution. By substitute $\mu_1,\mu_2,\dots,\mu_N$ with $\mu_1=\mu_2=\dots=\mu_N=0$ in (\ref{eq:KKT_FO}), we have
\begin{equation}\label{eq:KKT_FO2}
\left\{
\begin{aligned}
&\lambda_1 \left( 2+\gamma \sum\limits_{m=1}^{N}\theta_m \lambda_m \right) = W\log_2(e) \gamma\\
&\lambda_2 \left( 2+\gamma \sum\limits_{m=1}^{N}\theta_m \lambda_m \right) = W\log_2(e) \gamma\\
&\qquad\qquad\qquad \vdots\\
&\lambda_N \left( 2+\gamma \sum\limits_{m=1}^{N}\theta_m \lambda_m \right) = W\log_2(e) \gamma\\
\end{aligned}
\right.
\end{equation}
By solving the system of equations of (\ref{eq:KKT_FO2}), we can get a solution as given by
\begin{equation}\label{eq:KKT_sol2}
\lambda_1 = \lambda_2 = \dots = \lambda_N=\frac{\sqrt{\log_2(e)\gamma^2 W \Theta+1}-1}{\gamma \Theta},
\end{equation}
where $e$ is the base of the natural logarithm, $\Theta$ is given by
\begin{equation}\label{eq:Theta}
\Theta = \sum\limits_{m=1}^{N} \theta_m.
\end{equation}

Because problem (P5.3) is a convex optimization problem, we can conclude that this solution given by KKT conditions is the solution of original optimization problem.

\end{appendix}



\section*{Acknowledgment}
The authors would like to thank Prof. Zhu Han for his helpful discussion. The authors also thank the editor and anonymous reviewers for their valuable comments and suggestions, which improved the quality of the paper.

\ifCLASSOPTIONcaptionsoff
  \newpage
\fi

\bibliographystyle{IEEEtran}
\bibliography{EH_Contract_Theory}



\end{document}